\documentclass{notes}

\title{A model of navigation history}
\author{Connor G. Brewster \and Alan Jeffrey}
\date{August, 2016}
\indicia{
  \includegraphics[height=4.5ex]{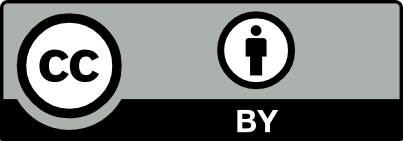}
  \begin{tabular}[b]{@{}l}%
    Copyright Connor G. Brewster and Alan Jeffrey\\
    Licensed under Creative Commons License CC-BY
  \end{tabular}
}

\usepackage{amssymb}
\usepackage{colortbl}
\usepackage{graphicx}
\usepackage{float}
\usepackage{subcaption}
\usepackage{tikz}
\usetikzlibrary{fit}

\newcommand{\Verts}{V}
\newcommand{\aVert}{v}
\newcommand{\bVert}{w}
\newcommand{\aNH}{H}
\newcommand{\Docs}{D}
\newcommand{\Active}{A}
\newcommand{\FullyActive}{F\!A}
\newcommand{\parentOf}{\rightarrow}
\newcommand{\parentOfActive}{\twoheadrightarrow}
\newcommand{\childOf}{\leftarrow}

\newcommand{\leChron}{\le}
\newcommand{\ltChron}{<}

\newcommand{\gtChron}{>}
\newcommand{\eqSess}{\sim}
\newcommand{\ltSess}{\lesssim}
\newcommand{\gtSess}{\gtrsim}
\newcommand{\rootDoc}{d_0}
\newcommand{\aDoc}{d}
\newcommand{\bDoc}{e}
\newcommand{\cDoc}{f}
\newcommand{\st}{\mathbin.}
\newtheorem{definition}{Definition}
\newtheorem{theorem}{Theorem}

\newtheorem{patch}{Patch}
\newtheorem{counterexample}{Counterexample}

\newcommand{\QED}{\hfill$\Box$}

\tikzstyle{doc} = [draw=black, fill=blue!10, circle, font={\normalfont\sffamily}]
\tikzstyle{fully} = [draw=red, thick]
\tikzstyle{active} = [color=white, fill=blue!50!black]
\tikzstyle{jshactive} = [active] 

\begin{document}

\maketitle

\subparagraph{Abstract:}
Navigation has been a core component of the web since its inception:
users and scripts can follow hyperlinks, and can go back or forwards
through the navigation history. In this paper, we present a formal
model aligned with the \textsc{whatwg} specification of navigation
history, and investigate its properties. The fundamental property of
navigation history is that traversing the history by $\delta$ then by
$\delta'$ should be the same as traversing by $\delta+\delta'$. In
particular, traversing by $+1$ (forward) then by $-1$ (back) is the
same as traversing by $0$ (doing nothing). We show that the
specification-aligned model does not satisfy this property, by
exhibiting a series of counter-examples, which motivate four patches
to the model. We present a series of experiments, showing that
browsers are inconsistent in their implementation of navigation
history, but that their behaviour is closer to the patched model than
to the specification-aligned model. We propose patches to the
specification to align it with the patched model.



\section{Introduction}

Navigation has been a core component of the web since its inception:
users and scripts can follow hyperlinks, and can go back or forwards
through the navigation history. Users are exposed to this functionality
through following hyperlinks, and by the forward and back buttons.
Scripts have many ways of accessing session history, via the
navigation API~\cite[\S7.7]{whatwg} and the \verb|element.click()| method.

Prior formalizations of navigation history
include~\cite{HH:2006,Haydar:2004,HPS:2004,LHYT:2000,WP:2003}, which
predate and are not aligned with the \textsc{whatwg}
specification~\cite{whatwg}.  The specification of the navigation API
is informal, and has complex dependencies on the rest of the HTML
specification. There is little discussion of the goals
of the API, and an unclear alignment with browser implementations.

In this paper, we present a formal model of navigation, aligned with
the HTML specification, and investigate its properties. The
starting point is that there is a total order of
\emph{documents}\footnote{%
  We are eliding some of the technical details of the specification here,
  in particular we are conflating a \emph{browsing context}
  with the document it contains, and we are ignoring issues around
  document loading and unloading, and around the current entry of the joint
  session history.
}, one of which is \emph{active}, for example:
\[\begin{tikzpicture}
  \node[doc](0) at (0,0){0};
  \node[doc](1) at (1,0){1};
  \node[doc,jshactive,fully](2) at (2,0){2};
  \node[draw,dotted,fit=(0)(1)(2)] {};
\end{tikzpicture}\]
In diagrams, we use left-to-right order to indicate order,
and highlight the active document. The user can \emph{traverse}
the history which changes the active document, for example pressing
the back button:
\[\begin{tikzpicture}
  \node[doc](0) at (0,0){0};
  \node[doc,jshactive,fully](1) at (1,0){1};
  \node[doc](2) at (2,0){2};
  \node[draw,dotted,fit=(0)(1)(2)] {};
\end{tikzpicture}\]
The user can also \emph{navigate}, which replaces any document
after the currently active document by a fresh active document:
\[\begin{tikzpicture}
  \node[doc](0) at (0,0){0};
  \node[doc](1) at (1,0){1};
  \node[doc,jshactive,fully](3) at (3,0){3};
  \node[draw,dotted,fit=(0)(1)(3)] {};
\end{tikzpicture}\]
Users can also traverse the history by more than one document
at a time, for example by using a pull-down menu from the back
or forwards button. This is called \emph{traversing by $\delta$},
for instance we can traverse our running example by $-2$
to get:
\[\begin{tikzpicture}
  \node[doc,jshactive,fully](0) at (0,0){0};
  \node[doc](1) at (1,0){1};
  \node[doc](3) at (3,0){3};
  \node[draw,dotted,fit=(0)(1)(3)] {};
\end{tikzpicture}\]
We formalize the notions of traversal and navigation in
\S\ref{sec:model}, and show the \emph{fundamental property of traversal}:
that traversing by $\delta$ then by $\delta'$
is the same as traversing by $\delta+\delta'$.

So far, the model is refreshingly simple, and corresponds well to
the specification and to browser implementations. Where the problems
arise is in the \emph{hierarchical} nature of documents. HTML
documents can contain \verb|iframe| elements, which
are independent documents in their own right, often
used to embed third party content such as advertisements.
We can treat each document as a tree, for example:
\[\begin{tikzpicture}
  \node[doc,jshactive,fully](0) at (0,0){0};
  \node[doc,active,fully](1) at (1,-1){1};
  \node[doc,active,fully](2) at (2,-2){2};
  \node[doc,active,fully](3) at (3,0){3};
  \node[draw,dotted,fit=(0)] {};
  \node[draw,dotted,fit=(1)] {};
  \node[draw,dotted,fit=(2)] {};
  \node[draw,dotted,fit=(3)] {};
  \draw[->](0)--(1);
  \draw[->](1)--(2);
  \draw[->](0)--(3);
\end{tikzpicture}\]
The problem comes from the ability of each document to
navigate separately and maintain its own session history,
but that traversal is a global operation that operates
on the \emph{joint session history}. For example
if document $2$ in the previous example navigates, the
resulting state is:
\[\begin{tikzpicture}
  \node[doc,active,fully](0) at (0,0){0};
  \node[doc,active,fully](1) at (1,-1){1};
  \node[doc](2) at (2,-2){2};
  \node[doc,active,fully](3) at (3,0){3};
  \node[doc,jshactive,fully](4) at (4,-2){4};
  \node[draw,dotted,fit=(0)] {};
  \node[draw,dotted,fit=(1)] {};
  \node[draw,dotted,fit=(2)(4)] {};
  \node[draw,dotted,fit=(3)] {};
  \draw[->](0)--(1);
  \draw[->](1)--(4);
  \draw[->](0)--(3);
\end{tikzpicture}\]
and then if document $1$ navigates, the state is:
\[\begin{tikzpicture}
  \node[doc,active,fully](0) at (0,0){0};
  \node[doc](1) at (1,-1){1};
  \node[doc](2) at (2,-2){2};
  \node[doc,active,fully](3) at (3,0){3};
  \node[doc,active](4) at (4,-2){4};
  \node[doc,jshactive,fully](5) at (5,-1){5};
  \node[draw,dotted,fit=(0)] {};
  \node[draw,dotted,fit=(1)(5)] {};
  \node[draw,dotted,fit=(2)(4)] {};
  \node[draw,dotted,fit=(3)] {};
  \draw[->](0)--(5);
  \draw[->](1)--(4);
  \draw[->](0)--(3);
\end{tikzpicture}\]
Note that node $4$ here is in an unusual state: it is active, but has
an inactive ancestor. The specification~\cite[\S7.7]{whatwg}
distinguishes between \emph{active} documents such as $4$, and
\emph{fully active} documents such as $0$, $3$ and $5$. Active
documents can become fully active by traversals involving their
ancestors. For example, after traversing by $-1$, document $4$ is
fully active:
\[\begin{tikzpicture}
  \node[doc,active,fully](0) at (0,0){0};
  \node[doc,jshactive,fully](1) at (1,-1){1};
  \node[doc](2) at (2,-2){2};
  \node[doc,active,fully](3) at (3,0){3};
  \node[doc,active,fully](4) at (4,-2){4};
  \node[doc](5) at (5,-1){5};
  \node[draw,dotted,fit=(0)] {};
  \node[draw,dotted,fit=(1)(5)] {};
  \node[draw,dotted,fit=(2)(4)] {};
  \node[draw,dotted,fit=(3)] {};
  \draw[->](0)--(1);
  \draw[->](1)--(4);
  \draw[->](0)--(3);
\end{tikzpicture}\]
As even a simple example like this shows, the combination of features
quickly results in a complex mix of session history, ordering, and
document hierarchy, which leads to the problems:
\begin{itemize}

\item \emph{Formally} there is no simple model,
  and the model provided by the specification does
  not satisfy the traverse-then-traverse property.

\item \emph{Experimentally} the browsers disagree
  with each other, and with the HTML specification,
  about the semantics of navigation.

\end{itemize}
In this paper, we address these:
\begin{itemize}

\item \S\ref{sec:model} provides a formal model of navigation history,
  which is intended to align with the specification. We show, through
  a series of examples, that it does not satisfy the
  fundamental property, and give patches to the model for
  each example. The final model does satisfy the
  fundamental property.

\item \S\ref{sec:experiments} investigates how well the patched
  model aligns with existing browser implementations. We show
  ways in which the browsers exhibit behaviours which are not
  aligned with the specification, and discuss how our proposed
  model matches these behaviours.

\end{itemize}
Finally, we propose changed wording to the specification, which
would bring it in line with our patched model.

\section{Model}
\label{sec:model}

In this section, we present our formal model of navigation history.
\S\ref{sec:preliminaries} contains definitions of concepts such as
tree and total order, and may be skipped by most readers. The model,
together with some examples, is given in \S\ref{sec:defns}. In
\S\ref{sec:properties} we define the fundamental property of
traversal, show that the model does \emph{not} satisfy
it, but can be patched to do so.

\subsection{Preliminaries}
\label{sec:preliminaries}

A \emph{directed graph} $G=(\Verts,{\parentOf})$ consists of:
\begin{itemize}
\item a set $\Verts$ (the \emph{vertices}), and
\item a relation ${\parentOf} \subseteq (\Verts\times\Verts)$ (the \emph{edges}).
\end{itemize}
The \emph{transitive closure} of $\parentOf$ is defined as $\aVert\parentOf^+\aVert'$ whenever
there exists $\aVert_0,\ldots,\aVert_n$ such that:
\[
  \aVert=\aVert_0\parentOf\cdots\parentOf\aVert_n=\aVert'
\]
The \emph{reflexive transitive closure} of $\parentOf$ is defined as $\aVert\parentOf^*\aVert'$ whenever
$\aVert\parentOf^+\aVert'$ or $\aVert=\aVert'$.
A \emph{forest} is a directed graph where:
\begin{itemize}
\item there is no $\aVert$ such that $\aVert\parentOf^+\aVert$ (\emph{acyclicity})
\item if $\aVert\parentOf\aVert'\childOf\aVert''$ then $\aVert=\aVert''$ (\emph{at most one parent}).
\end{itemize}
A \emph{root vertex} of a forest is a vertex $\aVert$ such that there is no $\bVert\parentOf\aVert$.
A \emph{tree} is a forest with a unique root vertex.
A \emph{preorder} is a directed graph $(\Verts, {\le})$ such that:
\begin{itemize}
\item every $\aVert$ has $\aVert\le\aVert$ (\emph{reflexivity}), and
\item if $\aVert\le\aVert'\le\aVert''$ then $\aVert\le\aVert''$ (\emph{transitivity}).
\end{itemize}
A \emph{partial order} is a preorder such that:
\begin{itemize}
\item for every $\aVert$ and $\aVert'$, if $\aVert\le\aVert'$ and $\aVert'\le\aVert$ then $\aVert=\aVert'$
  (\emph{antisymmetry}).
\end{itemize}
A \emph{total order} is a partial order such that:
\begin{itemize}
\item for every $\aVert$ and $\aVert'$, either $\aVert\le\aVert'$ or $\aVert\ge\aVert'$ (\emph{totality}).
\end{itemize}
A \emph{equivalence} is a preorder $(\Verts,{\sim})$ such that:
\begin{itemize}
\item if $\aVert\sim\aVert'$ then $\aVert'\sim\aVert$ (\emph{symmetry}).
\end{itemize}

\subsection{Definitions}
\label{sec:defns}

We can now formalize our model of navigation history, together with
the operations of navigation and traversal. This formalizes the
navigation history specification~\cite{whatwg}, and has a pleasant
diagrammatic presentation, but as we shall see in
\S\ref{sec:properties}, it has unexpected properties.

\begin{definition}[Navigation history]
A \emph{navigation history} $\aNH=(\Docs,\Active,{\parentOf},{\leChron},{\eqSess})$ consists of:
\begin{itemize}
\item a finite set $\Docs$ (the \emph{documents}),
\item a subset $\Active \subseteq \Docs$ (the \emph{active} documents),
\item a forest $(\Docs,{\parentOf})$ (the \emph{child} relationship),
\item a total order $(\Docs,{\leChron})$ (the \emph{chronological} order), and
\item an equivalence relation $(\Docs,{\eqSess})$ (the \emph{same-session} equivalence).
\end{itemize}
such that:
\begin{itemize}
\item for every $\aDoc$ there is a unique $\aDoc'\in\Active$ such that $\aDoc \eqSess \aDoc'$,
\item for every $\aDoc \parentOf \bDoc \eqSess \bDoc'$
  we have $\aDoc \parentOf \bDoc'$, and
\item for every $\aDoc \parentOf \bDoc$, we have $\aDoc \leChron \bDoc$.
  \QED
\end{itemize}
\end{definition}
We present such navigation histories ad diagrams, using
left-to-right position for chronological order, and grouping documents
in the same session. Since documents in the same session must have the
same parent, we only draw the document hierarchy for active children.
For example the diagram:
\[\begin{tikzpicture}
  \node[doc,active,fully](0) at (0,0){0};
  \node[doc,jshactive,fully](1) at (1,-1){1};
  \node[doc](2) at (2,-2){2};
  \node[doc,active,fully](3) at (3,0){3};
  \node[doc,active,fully](4) at (4,-2){4};
  \node[doc](5) at (5,-1){5};
  \node[draw,dotted,fit=(0)] {};
  \node[draw,dotted,fit=(1)(5)] {};
  \node[draw,dotted,fit=(2)(4)] {};
  \node[draw,dotted,fit=(3)] {};
  \draw[->](0)--(1);
  \draw[->](1)--(4);
  \draw[->](0)--(3);
\end{tikzpicture}\]
represents:
\[\begin{array}{l}
  D = \{ 0,1,2,3,4,5 \} \\[\jot]
  A = \{ 0,1,3,4 \} \\[\jot]
  0 \parentOf 1 \quad 0 \parentOf 3 \quad 0 \parentOf 5 \quad 1 \parentOf 2 \quad 1 \parentOf 4 \\[\jot]
  0 \leChron 1 \leChron 2 \leChron 3 \leChron 4 \leChron 5 \\[\jot]
  1 \eqSess 5 \quad 2 \eqSess 4
\end{array}\]
In such a navigation history, we define:
\begin{itemize}
\item $\rootDoc$ is the unique active root document,
\item $\aDoc \parentOfActive \bDoc$ when $\aDoc \parentOf \bDoc$ and $\bDoc \in \Active$
  (the \emph{active child} relationship),
\item $\FullyActive = \{ \aDoc \mid \rootDoc \parentOfActive^* \aDoc \}$
  (the \emph{fully active} documents),
\item $\aDoc \ltSess \bDoc$ whenever $\aDoc \eqSess \bDoc$ and $\aDoc \ltChron \bDoc$,
\item the \emph{session future} of $\aDoc$ is $\{ \bDoc \mid \aDoc \ltSess \bDoc \}$,
\item the \emph{session past} of $\aDoc$ is $\{ \bDoc \mid \aDoc \gtSess \bDoc \}$,
\item the \emph{joint session future} is $\{ \bDoc \mid \exists \aDoc \in \FullyActive \st \aDoc \ltSess \bDoc \}$,
\item the \emph{joint session past} is $\{ \bDoc \mid \exists \aDoc \in \FullyActive \st \aDoc \gtSess \bDoc \}$,
\end{itemize}
These definitions are intended to align with the specification, for example
\cite[7.7.2]{whatwg} has the definition:
\begin{quote}
  The \textbf{joint session history} of a top-level browsing context is the
  union of all the session histories of all browsing contexts of all
  the fully active \verb|Document| objects that share that top-level browsing
  context, with all the entries that are current entries in their
  respective session histories removed except for the current entry of
  the joint session history.
\end{quote}
which (eliding the concept of ``current entry of the joint session history'')
corresponds to the above definitions of joint session future and past.
We now consider how to formalize operations on navigation histories.
staring with \emph{navigation}, which is triggered by following hyperlinks,
or other actions which trigger document loading.

\begin{definition}[Navigation]
Define \emph{deleting $\aDoc$ from $\aNH$}, when $\aDoc$ is in the joint session future
to be $\aNH'=(\Docs',\Active,{\leChron},{\parentOf},{\eqSess})$ where:
\begin{itemize}
\item $\Docs' = \Docs \setminus \{ \bDoc \mid \aDoc\parentOf^* \bDoc \}$.
\end{itemize}
Define \emph{replacing $\aDoc$ by $\aDoc'$ in $\aNH$}, where $\aDoc\in\FullyActive$ and
$\aDoc'\notin\Docs$,
to be $\aNH'=(\Docs',\Active',{\leChron'},{\parentOf'},{\eqSess'})$ where:
\begin{itemize}
\item $\Docs' = \Docs \cup \{\aDoc'\}$,
\item $\bDoc \in \Active'$ whenever
  $\bDoc \in \Active$ and $\bDoc\ne\aDoc$, or
  $\bDoc=\aDoc'$,
\item $\bDoc \leChron' \cDoc$ whenever
  $\bDoc \leChron \cDoc$, or $\cDoc = \aDoc'$,
\item $\bDoc \parentOf' \cDoc$ whenever
  $\bDoc \parentOf \cDoc$, or
  $\bDoc \parentOf \aDoc$ and $\cDoc = \aDoc'$, and
\item $\bDoc \eqSess' \cDoc$ whenever
  $\bDoc \eqSess \cDoc$, or
  $\bDoc=\cDoc$, or
  $\bDoc \eqSess \aDoc$ and $\cDoc = \aDoc'$, or
  $\aDoc \eqSess \cDoc$ and $\bDoc = \aDoc'$.
\end{itemize}
Define \emph{navigating from $\aDoc$ to $\aDoc'$ in $\aNH$}, where $\aDoc\in\FullyActive$ to be the result of:
\begin{itemize}
\item deleting the session future of $\aDoc$, and then
\item replacing $\aDoc$ by $\aDoc'$.
  \QED
\end{itemize}
\end{definition}
There are two parts to navigation from $\aDoc$ to $\aDoc'$: deleting the session
future of $\aDoc$, followed by replacing $\aDoc$ by $\aDoc'$. For example,
navigating from $1$ to $6$ in:
\[\begin{tikzpicture}
  \node[doc,active,fully](0) at (0,0){0};
  \node[doc,jshactive,fully](1) at (1,-1){1};
  \node[doc](2) at (2,-2){2};
  \node[doc,active,fully](3) at (3,0){3};
  \node[doc,active,fully](4) at (4,-2){4};
  \node[doc](5) at (5,-1){5};
  \node[draw,dotted,fit=(0)] {};
  \node[draw,dotted,fit=(1)(5)] {};
  \node[draw,dotted,fit=(2)(4)] {};
  \node[draw,dotted,fit=(3)] {};
  \draw[->](0)--(1);
  \draw[->](1)--(4);
  \draw[->](0)--(3);
\end{tikzpicture}\]
we first delete the session future of $1$ (which is $5$):
\[\begin{tikzpicture}
  \node[doc,active,fully](0) at (0,0){0};
  \node[doc,jshactive,fully](1) at (1,-1){1};
  \node[doc](2) at (2,-2){2};
  \node[doc,active,fully](3) at (3,0){3};
  \node[doc,active,fully](4) at (4,-2){4};
  \node[draw,dotted,fit=(0)] {};
  \node[draw,dotted,fit=(1)] {};
  \node[draw,dotted,fit=(2)(4)] {};
  \node[draw,dotted,fit=(3)] {};
  \draw[->](0)--(1);
  \draw[->](1)--(4);
  \draw[->](0)--(3);
\end{tikzpicture}\]
then replace $1$ by $6$:
\[\begin{tikzpicture}
  \node[doc,active,fully](0) at (0,0){0};
  \node[doc](1) at (1,-1){1};
  \node[doc](2) at (2,-2){2};
  \node[doc,active,fully](3) at (3,0){3};
  \node[doc,active](4) at (4,-2){4};
  \node[doc,jshactive,fully](6) at (6,-1){6};
  \node[draw,dotted,fit=(0)] {};
  \node[draw,dotted,fit=(1)(6)] {};
  \node[draw,dotted,fit=(2)(4)] {};
  \node[draw,dotted,fit=(3)] {};
  \draw[->](0)--(6);
  \draw[->](1)--(4);
  \draw[->](0)--(3);
\end{tikzpicture}\]
We also define \emph{traversing the history}, which changes the active
documents.
\begin{definition}[Traversal]
Define \emph{traversing the history to $\aDoc$ in $\aNH$}, where $\aDoc \in \Docs$,
to be $\aNH'=(\Docs,\Active',{\leChron},{\parentOf},{\eqSess})$ where:
\begin{itemize}
\item $\bDoc\in\Active'$ whenever $\aDoc\not\eqSess\bDoc \in \Active$, or
  $\aDoc=\bDoc$.
\end{itemize}
Define \emph{$\aNH$ traverses the history by $+\delta$ to $\aNH'$} when:
\begin{itemize}
\item the joint session future of $\aNH$ is $\aDoc_1 \ltChron \cdots \ltChron \aDoc_\delta \ltChron \cdots$,
\item $H$ traverses the history to $d_\delta$ in $H'$
\end{itemize}
Define \emph{$\aNH$ traverses the history by $-\delta$ to $\aNH'$} when:
\begin{itemize}
\item the joint session past of $\aNH$ is $\aDoc_1 \gtChron \cdots \gtChron \aDoc_\delta \gtChron \cdots$,
\item $H$ traverses the history to $d_\delta$ in $H'$
\end{itemize}
Define \emph{$\aNH$ traverses the history by $0$ to $\aNH'$} when $\aNH=\aNH'$.
\end{definition}
For example, to traverse the history by $-2$ in:
\[\begin{tikzpicture}
  \node[doc,active,fully](0) at (0,0){0};
  \node[doc](1) at (1,-1){1};
  \node[doc](2) at (2,-2){2};
  \node[doc,active,fully](3) at (3,0){3};
  \node[doc,active](4) at (4,-2){4};
  \node[doc,jshactive,fully](6) at (6,-1){6};
  \node[draw,dotted,fit=(0)] {};
  \node[draw,dotted,fit=(1)(6)] {};
  \node[draw,dotted,fit=(2)(4)] {};
  \node[draw,dotted,fit=(3)] {};
  \draw[->](0)--(6);
  \draw[->](1)--(4);
  \draw[->](0)--(3);
\end{tikzpicture}\]
we find the joint session past (which is $2 \gtChron 1$)
and traverse the history to the second item (which is $1$)
to arrive at:
\[\begin{tikzpicture}
  \node[doc,active,fully](0) at (0,0){0};
  \node[doc,jshactive,fully](1) at (1,-1){1};
  \node[doc](2) at (2,-2){2};
  \node[doc,active,fully](3) at (3,0){3};
  \node[doc,active,fully](4) at (4,-2){4};
  \node[doc](6) at (6,-1){6};
  \node[draw,dotted,fit=(0)] {};
  \node[draw,dotted,fit=(1)(6)] {};
  \node[draw,dotted,fit=(2)(4)] {};
  \node[draw,dotted,fit=(3)] {};
  \draw[->](0)--(1);
  \draw[->](1)--(4);
  \draw[->](0)--(3);
\end{tikzpicture}\]
These definitions are intended to formally capture the HTML
specification, for example \cite[\S7.7.2]{whatwg} includes:
\begin{quote}
  To \textbf{traverse the history by a delta} $\delta$, the user agent
  must append a task to this top-level browsing context's session
  history traversal queue, the task consisting of running the
  following steps:
  \begin{enumerate}

  \item If the index of the current entry of the joint session history
    plus $\delta$ is less than zero or greater than or equal to the
    number of items in the joint session history, then abort these
    steps.

  \item Let \emph{specified entry} be the entry in the joint session
    history whose index is the sum of $\delta$ and the index of the
    current entry of the joint session history.

  \item Let \emph{specified browsing context} be the browsing context
    of the specified entry.

  \item If the specified browsing context's active document's unload a
    document algorithm is currently running, abort these steps.

  \item Queue a task that consists of running the following
    substeps [\dots]

    \begin{itemize}

    \item[3.] Traverse the history of the specified browsing context
      to the specified entry.

    \end{itemize}
  \end{enumerate}
\end{quote}

\subsection{Properties}
\label{sec:properties}

We now consider the fundamental property of navigation history:

\begin{definition}[Fundamental property]
\label{defn:fundamental}
  $H$ satisfies the \emph{fundamental property of traversal} whenever
  $H$ traverses the history by $\delta$ to $H'$
  and $H'$ traverses the history by $\delta'$ to $H''$
  implies $H$ traverses the history by $\delta+\delta'$ to $H''$.
  \QED
\end{definition}
Unfortunately, navigation histories as specified do not always satisfy the fundamental property,
due to ways individual session histories are combined into the joint session history.
In this section, we give a series of counterexamples, and propose patches to
the model to address each counterexample.

\begin{counterexample}
  \label{counterexample:intermediaries}
  Let $H$ be:
  \[\begin{tikzpicture}
    \node[doc,active,fully](0) at (0,0){0};
    \node[doc,jshactive,fully](1) at (1,-2){1};
    \node[doc](3) at (3,-2){3};
    \node[doc,active,fully](2) at (2,-1){2};
    \node[doc](4) at (4,-1){4};
    \node[draw,dotted,fit=(0)] {};
    \node[draw,dotted,fit=(1)(3)] {};
    \node[draw,dotted,fit=(2)(4)] {};
    \draw[->](0)--(1);
    \draw[->](0)--(2);
  \end{tikzpicture}\]
  which traverses the history by $1$ to:
  \[\begin{tikzpicture}
    \node[doc,active,fully](0) at (0,0){0};
    \node[doc](1) at (1,-2){1};
    \node[doc,jshactive,fully](3) at (3,-2){3};
    \node[doc,active,fully](2) at (2,-1){2};
    \node[doc](4) at (4,-1){4};
    \node[draw,dotted,fit=(0)] {};
    \node[draw,dotted,fit=(1)(3)] {};
    \node[draw,dotted,fit=(2)(4)] {};
    \draw[->](0)to[out=300,in=180](3);
    \draw[->](0)--(2);
  \end{tikzpicture}\]
  which traverses the history by $1$ to:
  \[\begin{tikzpicture}
    \node[doc,active,fully](0) at (0,0){0};
    \node[doc](1) at (1,-2){1};
    \node[doc,active,fully](3) at (3,-2){3};
    \node[doc](2) at (2,-1){2};
    \node[doc,jshactive,fully](4) at (4,-1){4};
    \node[draw,dotted,fit=(0)] {};
    \node[draw,dotted,fit=(1)(3)] {};
    \node[draw,dotted,fit=(2)(4)] {};
    \draw[->](0)to[out=300,in=180](3);
    \draw[->](0)--(4);
  \end{tikzpicture}\]
  but $H$ traverses the history by $2$ to:
  \[\begin{tikzpicture}
    \node[doc,active,fully](0) at (0,0){0};
    \node[doc,active,fully](1) at (1,-2){1};
    \node[doc](3) at (3,-2){3};
    \node[doc](2) at (2,-1){2};
    \node[doc,jshactive,fully](4) at (4,-1){4};
    \node[draw,dotted,fit=(0)] {};
    \node[draw,dotted,fit=(1)(3)] {};
    \node[draw,dotted,fit=(2)(4)] {};
    \draw[->](0)--(1);
    \draw[->](0)--(4);
  \end{tikzpicture}\]
  \QED
\end{counterexample}
This counterexample is caused by the definition of `traverses the history by $\delta$' which
only traverses one document's session history. Instead, we should traverse
the history of all $\delta$ documents.

\begin{patch}[Traverse intermediaries]
Define \emph{$\aNH$ traverses the history by $+\delta$ to $\aNH'$} when:
\begin{itemize}
\item the joint session future of $\aNH$ is $\aDoc_1 \ltChron \cdots \ltChron \aDoc_\delta \ltChron \cdots$,
\item there is some $\aNH=\aNH_0,\ldots,\aNH_\delta=\aNH'$, such that
\item $H_{i-1}$ traverses the history to $d_i$ in $H_i$ for each $1 \le i \le \delta$.
\end{itemize}
Define \emph{$\aNH$ traverses the history by $-\delta$ to $\aNH'$} when:
\begin{itemize}
\item the joint session past of $\aNH$ is $\aDoc_1 \gtChron \cdots \gtChron \aDoc_\delta \gtChron \cdots$,
\item there is some $\aNH=\aNH_0,\ldots,\aNH_\delta=\aNH'$, such that
\item $H_{i-1}$ traverses the history to $d_i$ in $H_i$ for each $1 \le i \le \delta$.
  \QED
\end{itemize}
\end{patch}

\begin{counterexample}
  Let $H$ be:
  \[\begin{tikzpicture}
    \node[doc,active,fully](0) at (0,0){0};
    \node[doc,jshactive,fully](1) at (1,-1){1};
    \node[doc](2) at (2,-1){2};
    \node[doc,active](3) at (3,-2){3};
    \node[doc](4) at (4,-2){4};
    \node[doc](5) at (5,0){5};
    \node[draw,dotted,fit=(0)(5)] {};
    \node[draw,dotted,fit=(1)(2)] {};
    \node[draw,dotted,fit=(3)(4)] {};
    \draw[->](0)--(1);
    \draw[->](2)--(3);
  \end{tikzpicture}\]
  which traverses the history by $1$ to:
  \[\begin{tikzpicture}
    \node[doc,active,fully](0) at (0,0){0};
    \node[doc](1) at (1,-1){1};
    \node[doc,jshactive,fully](2) at (2,-1){2};
    \node[doc,active,fully](3) at (3,-2){3};
    \node[doc](4) at (4,-2){4};
    \node[doc](5) at (5,0){5};
    \node[draw,dotted,fit=(0)(5)] {};
    \node[draw,dotted,fit=(1)(2)] {};
    \node[draw,dotted,fit=(3)(4)] {};
    \draw[->](0)--(2);
    \draw[->](2)--(3);
  \end{tikzpicture}\]
  which in turn traverses the history by $1$ to:
  \[\begin{tikzpicture}
    \node[doc,active,fully](0) at (0,0){0};
    \node[doc](1) at (1,-1){1};
    \node[doc,active,fully](2) at (2,-1){2};
    \node[doc](3) at (3,-2){3};
    \node[doc,jshactive,fully](4) at (4,-2){4};
    \node[doc](5) at (5,0){5};
    \node[draw,dotted,fit=(0)(5)] {};
    \node[draw,dotted,fit=(1)(2)] {};
    \node[draw,dotted,fit=(3)(4)] {};
    \draw[->](0)--(2);
    \draw[->](2)--(4);
  \end{tikzpicture}\]
  but $H$ traverses the history by $2$ to:
  \[\begin{tikzpicture}
    \node[doc](0) at (0,0){0};
    \node[doc](1) at (1,-1){1};
    \node[doc,active](2) at (2,-1){2};
    \node[doc,active](3) at (3,-2){3};
    \node[doc](4) at (4,-2){4};
    \node[doc,jshactive,fully](5) at (5,0){5};
    \node[draw,dotted,fit=(0)(5)] {};
    \node[draw,dotted,fit=(1)(2)] {};
    \node[draw,dotted,fit=(3)(4)] {};
    \draw[->](0)--(2);
    \draw[->](2)--(3);
  \end{tikzpicture}\]
  \QED
\end{counterexample}
The problem this time is that the definition of `joint session history' only includes
the fully active documents, not all active documents.

\begin{patch}[Active joint session history]
Define:
\begin{itemize}
\item the \emph{joint session future} is $\{ \bDoc \mid \exists \aDoc \in \Active \st \aDoc \ltSess \bDoc \}$, and
\item the \emph{joint session past} is $\{ \bDoc \mid \exists \aDoc \in \Active \st \aDoc \gtSess \bDoc \}$.
  \QED
\end{itemize}
\end{patch}

\begin{counterexample}
  Let $H$ be:
  \[\begin{tikzpicture}
    \node[doc,active,fully](0) at (0,0){0};
    \node[doc](1) at (1,-1){1};
    \node[doc,jshactive,fully](2) at (2,-2){2};
    \node[doc](3) at (3,-2){3};
    \node[doc,active,fully](4) at (4,-1){4};
    \node[draw,dotted,fit=(0)]{};
    \node[draw,dotted,fit=(1)(4)]{};
    \node[draw,dotted,fit=(2)(3)]{};
    \draw[->](0)--(4);
    \draw[->](0)to[out=-20,in=90](2);
  \end{tikzpicture}\]
  which traverses the history by $-1$ to:
  \[\begin{tikzpicture}
    \node[doc,active,fully](0) at (0,0){0};
    \node[doc,jshactive,fully](1) at (1,-1){1};
    \node[doc,active,fully](2) at (2,-2){2};
    \node[doc](3) at (3,-2){3};
    \node[doc](4) at (4,-1){4};
    \node[draw,dotted,fit=(0)]{};
    \node[draw,dotted,fit=(1)(4)]{};
    \node[draw,dotted,fit=(2)(3)]{};
    \draw[->](0)--(1);
    \draw[->](0)to[out=-20,in=90](2);
  \end{tikzpicture}\]
  which traverses the history by $1$ to:
  \[\begin{tikzpicture}
    \node[doc,active,fully](0) at (0,0){0};
    \node[doc,active,fully](1) at (1,-1){1};
    \node[doc](2) at (2,-2){2};
    \node[doc,jshactive,fully](3) at (3,-2){3};
    \node[doc](4) at (4,-1){4};
    \node[draw,dotted,fit=(0)]{};
    \node[draw,dotted,fit=(1)(4)]{};
    \node[draw,dotted,fit=(2)(3)]{};
    \draw[->](0)--(1);
    \draw[->](0)to[out=-20,in=120](3);
  \end{tikzpicture}\]
  which is not the same as $H$.
  \QED
\end{counterexample}
This counterexample is caused by an asymmetry in the definition
of traversal: it is defined in terms of navigating \emph{to} a document
$d$, and not navigating \emph{from} a document. We fix this
by making the definition symmetric:

\begin{patch}[Symmetric traversal]
Define \emph{$\aNH$ traverses the history from $\aDoc'$} when there is some $\aDoc$ such that:
\begin{itemize}
\item $\aDoc\ltSess\aDoc'$,
\item for any $\bDoc\ltSess\aDoc'$ we have $\bDoc\leChron\aDoc$, and
\item $\aNH$ traverses the history to $\aDoc$.
\end{itemize}
Define \emph{$\aNH$ traverses the history by $-\delta$ to $\aNH'$} when:
\begin{itemize}
\item the joint session past and active documents of $\aNH$ are $\aDoc_1 \gtChron \cdots \gtChron \aDoc_\delta \gtChron \cdots$,
\item there is some $\aNH=\aNH_0,\ldots,\aNH_\delta=\aNH'$, such that
\item $H_{i-1}$ traverses the history from $d_i$ in $H_i$ for each $1 \le i \le \delta$.
  \QED
\end{itemize}
\end{patch}
For example, to traverse the history by $-1$ from:
\[\begin{tikzpicture}
\node[doc,active,fully](0) at (0,0){0};
\node[doc](1) at (1,-1){1};
\node[doc,jshactive,fully](2) at (2,-2){2};
\node[doc](3) at (3,-2){3};
\node[doc,active,fully](4) at (4,-1){4};
\node[draw,dotted,fit=(0)]{};
\node[draw,dotted,fit=(1)(4)]{};
\node[draw,dotted,fit=(2)(3)]{};
\draw[->](0)--(4);
\draw[->](0)to[out=-20,in=90](2);
\end{tikzpicture}\]
we find the joint session past and active documents (which is $4 \gtChron 2 \gtChron 1 \gtChron 0$)
and traverse the history from the first item (which is $4$)
which is the same as traversing the history to $1$:
\[\begin{tikzpicture}
\node[doc,active,fully](0) at (0,0){0};
\node[doc,jshactive,fully](1) at (1,-1){1};
\node[doc,active,fully](2) at (2,-2){2};
\node[doc](3) at (3,-2){3};
\node[doc](4) at (4,-1){4};
\node[draw,dotted,fit=(0)]{};
\node[draw,dotted,fit=(1)(4)]{};
\node[draw,dotted,fit=(2)(3)]{};
\draw[->](0)--(1);
\draw[->](0)to[out=-20,in=90](2);
\end{tikzpicture}\]

\begin{counterexample}
\label{cex:not-well-formed}
  Let $H$ be:
  \[\begin{tikzpicture}
    \node[doc,active,fully](0) at (0,0){0};
    \node[doc,active,fully](1) at (1,-1){1};
    \node[doc](2) at (2,-2){2};
    \node[doc](3) at (3,-1){3};
    \node[doc,jshactive,fully](4) at (4,-2){4};
    \node[draw,dotted,fit=(0)]{};
    \node[draw,dotted,fit=(1)(3)]{};
    \node[draw,dotted,fit=(2)(4)]{};
    \draw[->](0)--(1);
    \draw[->](1)--(4);
  \end{tikzpicture}\]
  which traverses the history by $-1$ to:
  \[\begin{tikzpicture}
    \node[doc,active,fully](0) at (0,0){0};
    \node[doc,active,fully](1) at (1,-1){1};
    \node[doc,jshactive,fully](2) at (2,-2){2};
    \node[doc](3) at (3,-1){3};
    \node[doc ](4) at (4,-2){4};
    \node[draw,dotted,fit=(0)]{};
    \node[draw,dotted,fit=(1)(3)]{};
    \node[draw,dotted,fit=(2)(4)]{};
    \draw[->](0)--(1);
    \draw[->](1)--(2);
  \end{tikzpicture}\]
  which traverses the history by $1$ to:
  \[\begin{tikzpicture}
    \node[doc,active,fully](0) at (0,0){0};
    \node[doc](1) at (1,-1){1};
    \node[doc,active](2) at (2,-2){2};
    \node[doc,jshactive,fully](3) at (3,-1){3};
    \node[doc](4) at (4,-2){4};
    \node[draw,dotted,fit=(0)]{};
    \node[draw,dotted,fit=(1)(3)]{};
    \node[draw,dotted,fit=(2)(4)]{};
    \draw[->](0)--(3);
    \draw[->](1)--(2);
  \end{tikzpicture}\]
  which is not the same as $H$.
  \QED
\end{counterexample}
The problem here is not the definition of `traversing by $\delta$', but the definition
of navigation histories themselves. They allow for states such as $H$ from
Counterexample~\ref{cex:not-well-formed}, which includes the problematic documents:
\[\begin{tikzpicture}
  \node[doc,active,fully](1) at (1,-1){1};
  \node[doc](2) at (2,-2){2};
  \node[doc](3) at (3,-1){3};
  \node[doc,jshactive,fully](4) at (4,-2){4};
  \node[draw,dotted,fit=(1)(3)]{};
  \node[draw,dotted,fit=(2)(4)]{};
\end{tikzpicture}\]
There are similar problems with documents:
\[\begin{tikzpicture}
  \node[doc,active,fully](2) at (2,-1){2};
  \node[doc](1) at (1,-2){1};
  \node[doc](3) at (3,-1){3};
  \node[doc,jshactive,fully](4) at (4,-2){4};
  \node[draw,dotted,fit=(2)(3)]{};
  \node[draw,dotted,fit=(1)(4)]{};
\end{tikzpicture}\]
and with documents:
\[\begin{tikzpicture}
  \node[doc,active,fully](1) at (1,-1){1};
  \node[doc](3) at (3,-2){3};
  \node[doc](2) at (2,-1){2};
  \node[doc,jshactive,fully](4) at (4,-2){4};
  \node[draw,dotted,fit=(1)(2)]{};
  \node[draw,dotted,fit=(3)(4)]{};
\end{tikzpicture}\]
It turns out that these are the only remaining cause of counterexamples,
and we will call examples like this not \emph{well-formed}.

\begin{definition}[Well-formed]
  A navigation history is \emph{well formed} whenever
  for any $a \ltSess b$ and $c \ltSess d$,
  if $a \in \Active$ and $d \in \Active$ then $d \leChron b$.
\end{definition}
We have that traversal preserves being well-formed: if $H$ is well-formed, and $H$ traverses
by $\delta$ to $H'$, then $H'$ is well-formed. Unfortunately, this is not true for navigation,
because of the way it clears the session future.

\begin{counterexample}
\label{cex:wf-nav}
  Let $H$ be the well-formed history:
  \[\begin{tikzpicture}
    \node[doc,active,fully](0) at (0,0){0};
    \node[doc,active,fully](1) at (1,-1){1};
    \node[doc,jshactive,fully](2) at (2,-2){2};
    \node[doc](3) at (3,-1){3};
    \node[draw,dotted,fit=(0)]{};
    \node[draw,dotted,fit=(1)(3)]{};
    \node[draw,dotted,fit=(2)]{};
    \draw[->](0)--(1);
    \draw[->](1)--(2);
  \end{tikzpicture}\]
  which navigates from $2$ to:
  \[\begin{tikzpicture}
    \node[doc,active,fully](0) at (0,0){0};
    \node[doc,active,fully](1) at (1,-1){1};
    \node[doc](2) at (2,-2){2};
    \node[doc](3) at (3,-1){3};
    \node[doc,jshactive,fully](4) at (4,-2){4};
    \node[draw,dotted,fit=(0)]{};
    \node[draw,dotted,fit=(1)(3)]{};
    \node[draw,dotted,fit=(2)(4)]{};
    \draw[->](0)--(1);
    \draw[->](1)--(4);
  \end{tikzpicture}\]
  which is not well-formed.
  \QED
\end{counterexample}
Fortunately, we can patch navigation to address this, by requiring that
we clear the entire joint session future, not just the session future of the document
being navigated from.

\begin{patch}[Navigation deletes joint session future]
Define \emph{navigating from $\aDoc$ to $\aDoc'$ in $\aNH$}, where $\aDoc\in\FullyActive$ to be the result of:
\begin{itemize}
\item deleting the joint session future, and then
\item replacing $\aDoc$ by $\aDoc'$.
  \QED
\end{itemize}
\end{patch}
For example in Counterexample~\ref{cex:wf-nav}, navigation from 2 now results in the well-formed history:
  \[\begin{tikzpicture}
    \node[doc,active,fully](0) at (0,0){0};
    \node[doc,active,fully](1) at (1,-1){1};
    \node[doc](2) at (2,-2){2};
    \node[doc,jshactive,fully](4) at (4,-2){4};
    \node[draw,dotted,fit=(0)]{};
    \node[draw,dotted,fit=(1)]{};
    \node[draw,dotted,fit=(2)(4)]{};
    \draw[->](0)--(1);
    \draw[->](1)--(4);
  \end{tikzpicture}\]
With these patches, we can prove the fundamental property of traversal.

\begin{theorem}
\label{thm:fundamental}
  For any well-formed navigation history $H$,
  if $H$ traverses the history by $\delta$ to $H'$
  and $H'$ traverses the history by $\delta'$ to $H''$
  then $H$ traverses the history by $\delta+\delta'$ to $H''$.
\end{theorem}
\begin{proof}
  In this paper, we give a proof sketch. The full details have been mechanically verified in Agda~\cite{AgdaProofs}.
  Define:
  \begin{itemize}
  \item a document $d$ \emph{can go back} there is some $c \ltSess d$,
  \item the \emph{back target} $b$ is the $\le$-largest active document which can go back, and
  \item the \emph{forward target} $f$ is the $\le$-smallest document in the joint session future.
  \end{itemize}
  We then show some lemmas:
  \begin{enumerate}
  \item $H$ traverses by $+(\delta+1)$ to $H'$ if and only if
    $H$ traverses to $f$, then traverses by $+\delta$ to $H'$.
  \item $H$ traverses by $-(\delta+1)$ to $H'$ if and only if
    $H$ traverses from $b$, then traverses by $-\delta$ to $H'$.
  \item If $H$ is well-formed and $H$ traverses to $f$ with result $H'$,
    then $f$ is the back target of $H'$, and $H'$ traverses from $f$ with result $H$.
  \item If $H$ is well-formed and $H$ traverses from $b$ with result $H'$,
    then $b$ is the forward target of $H'$, and $H'$ traverses to $b$ with result $H$.
  \item If $H$ is well-formed and $H$ traverses to $f$ to $H'$, then $H'$ is well-formed.
  \item If $H$ is well-formed and $H$ traverses from $b$ to $H'$, then $H'$ is well-formed.
  \end{enumerate}
  The result is then an induction on $\delta$.
  \QED
\end{proof}

\section{Experiments}
\label{sec:experiments}

In this section, we summarize our experiments to validate the conformance of browser
implementations with respect to the \textsc{whatwg} specification, to our proposed
changes, and to each other.

We give details of how to recreate Counterexample~\ref{counterexample:intermediaries}
in detail, the other counterexamples are similar. We create an \textsc{html} page for the parent,
containing two \verb|iframe|s, both of which start at \verb|page1.html|, with a hyperlink
to \verb|page2.html|:
\begin{quote}
\raisebox{-.5\height}{
   \includegraphics[width=.45\linewidth]{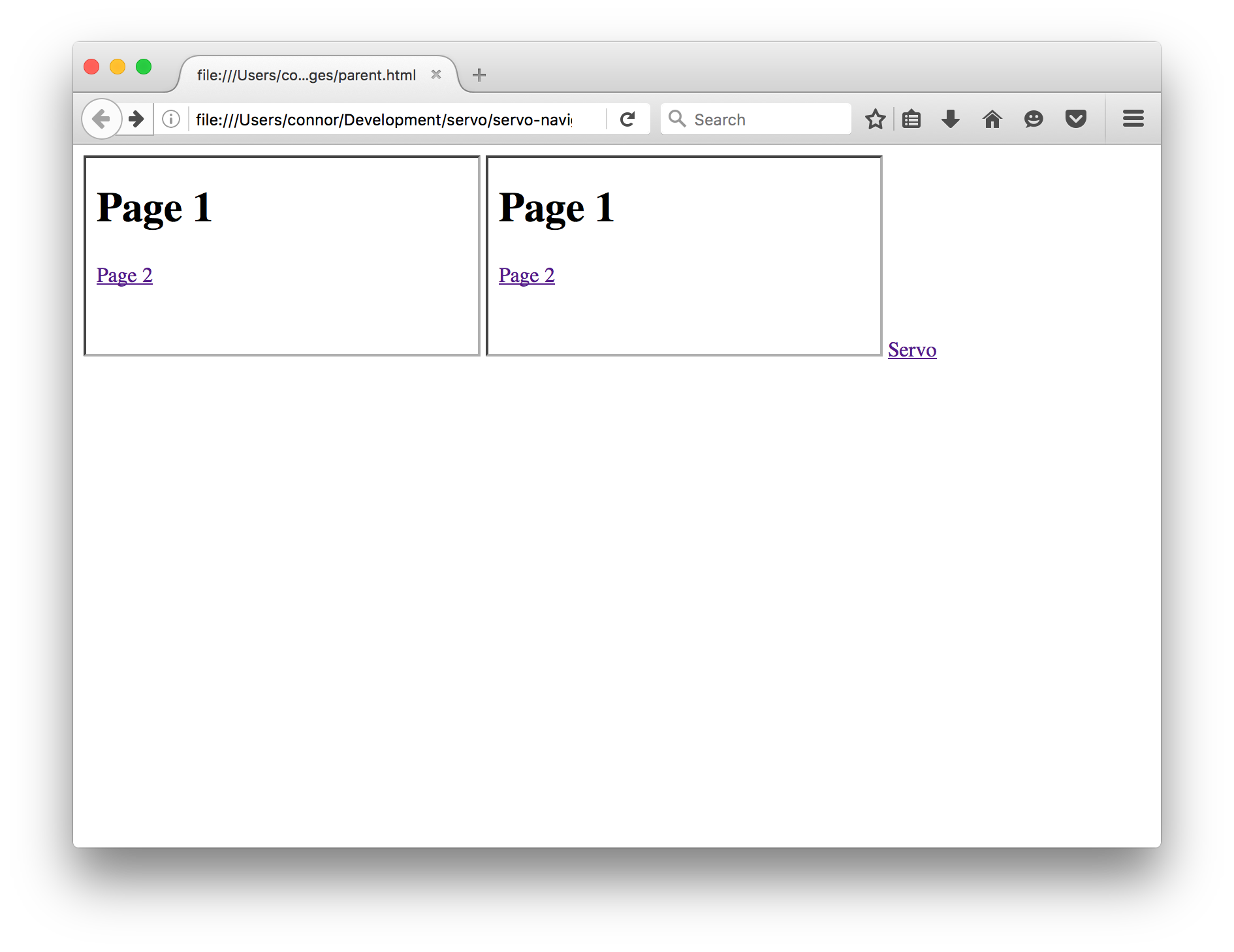}%
}~\raisebox{-.5\height}{
  \begin{tikzpicture}
    \node[doc,active,fully](0) at (0,0){0};
    \node[doc,jshactive,fully](1) at (1,-2){1};
    \node[doc,active,fully](2) at (2,-1){2};
    \node[draw,dotted,fit=(0)] {};
    \node[draw,dotted,fit=(1)] {};
    \node[draw,dotted,fit=(2)] {};
    \draw[->](0)--(1);
    \draw[->](0)--(2);
  \end{tikzpicture}
}
\end{quote}
Clicking on both hyperlinks loads both copies of \verb|page2.html|:
\begin{quote}
\raisebox{-.5\height}{
   \includegraphics[width=.45\linewidth]{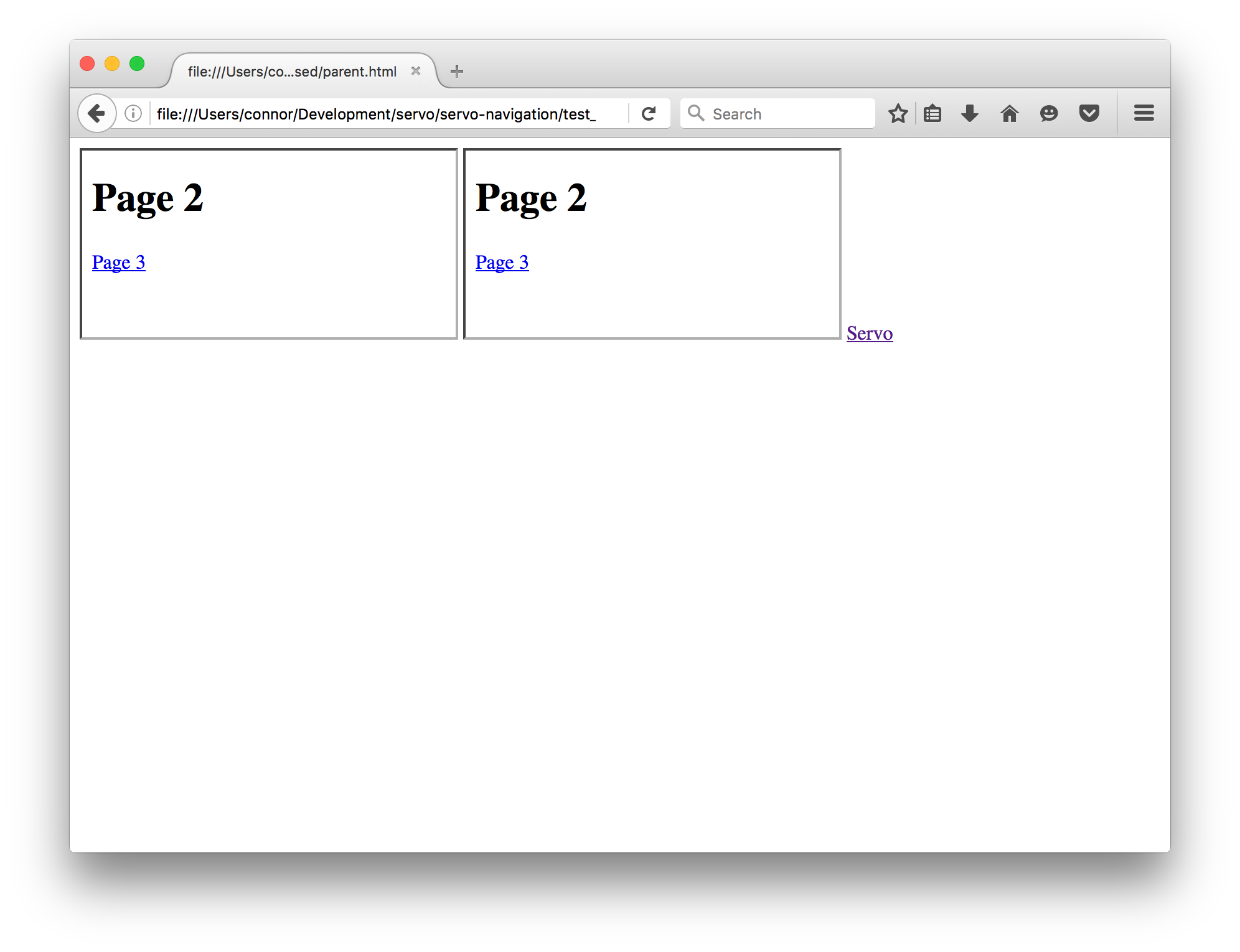}%
}~\raisebox{-.5\height}{
  \begin{tikzpicture}
    \node[doc,active,fully](0) at (0,0){0};
    \node[doc](1) at (1,-2){1};
    \node[doc,active,fully](3) at (3,-2){3};
    \node[doc](2) at (2,-1){2};
    \node[doc,jshactive,fully](4) at (4,-1){4};
    \node[draw,dotted,fit=(0)] {};
    \node[draw,dotted,fit=(1)(3)] {};
    \node[draw,dotted,fit=(2)(4)] {};
    \draw[->](0)to[out=300,in=180](3);
    \draw[->](0)--(4);
  \end{tikzpicture}
}
\end{quote}
Pressing the ``back'' button twice takes us to the initial state of Counterexample~\ref{counterexample:intermediaries}:
\begin{quote}
\raisebox{-.5\height}{
   \includegraphics[width=.45\linewidth]{images/experiments/forwardback4/firefox/1.png}%
}~\raisebox{-.5\height}{
  \begin{tikzpicture}
    \node[doc,active,fully](0) at (0,0){0};
    \node[doc,jshactive,fully](1) at (1,-2){1};
    \node[doc](3) at (3,-2){3};
    \node[doc,active,fully](2) at (2,-1){2};
    \node[doc](4) at (4,-1){4};
    \node[draw,dotted,fit=(0)] {};
    \node[draw,dotted,fit=(1)(3)] {};
    \node[draw,dotted,fit=(2)(4)] {};
    \draw[->](0)--(1);
    \draw[->](0)--(2);
  \end{tikzpicture}
}
\end{quote}
Now, the user can traverse the history by $+2$ (by holding down the ``forward'' button)
which results in state:
\begin{quote}
\raisebox{-.5\height}{
   \includegraphics[width=.45\linewidth]{images/experiments/forwardback4/firefox/8.png}%
}~\raisebox{-.5\height}{
  \begin{tikzpicture}
    \node[doc,active,fully](0) at (0,0){0};
    \node[doc](1) at (1,-2){1};
    \node[doc,active,fully](3) at (3,-2){3};
    \node[doc](2) at (2,-1){2};
    \node[doc,jshactive,fully](4) at (4,-1){4};
    \node[draw,dotted,fit=(0)] {};
    \node[draw,dotted,fit=(1)(3)] {};
    \node[draw,dotted,fit=(2)(4)] {};
    \draw[->](0)to[out=300,in=180](3);
    \draw[->](0)--(4);
  \end{tikzpicture}
}
\end{quote}
Experimentally, this shows that Firefox is aligned with our patched model, rather than
with the unpatched model. We can set up similar experiments for the other counterexamples,
and execute them in other browsers, which gives results\footnote{%
  Recall that Counterexample~4 depends on a non-well-formed navigation history,
  and that the patch for it is to make such states unreachable, and so
  experimentally unverifiable.
}:
\begin{center}
 {\sffamily
  \begin{tabular}{crrrr}
    \rowcolor{black!50!blue}
    &&&& \llap{\color{white} Counterexample} \\
    \rowcolor{black!50!blue}
    & \color{white}1 & \color{white}2 & \color{white}3 & \color{white}5 \\
    \rowcolor{white!90!blue}
    Firefox           & P  & P  & P  & P \\
    Chrome            & P  & P  & P  & P \\
    \rowcolor{white!90!blue}
    Safari            & P  & P  & P  & P \\
    Internet Explorer & U  & U  & P  & P \\
    \rowcolor{white!90!blue}
    Servo             & P  & P  & P  & P \\
  \end{tabular}
 }
 \quad
 \begin{tabular}{ll}
   \textsf{P}:& aligned with patched model \\
   \textsf{U}:& aligned with unpatched model \\
 \end{tabular}
\end{center}
Most browsers are compatible with the patched model rather than than
unpatched model, with the exception of Internet Explorer, which has mixed behaviour
(Edge is similar). Servo was designed from the patched model.

Moreover, performing these experiments shows some unexpected behaviours
in browser implementations. For example in Firefox, starting in state:
\begin{quote}
    \raisebox{-.5\height}{
      \includegraphics[width=.45\linewidth]{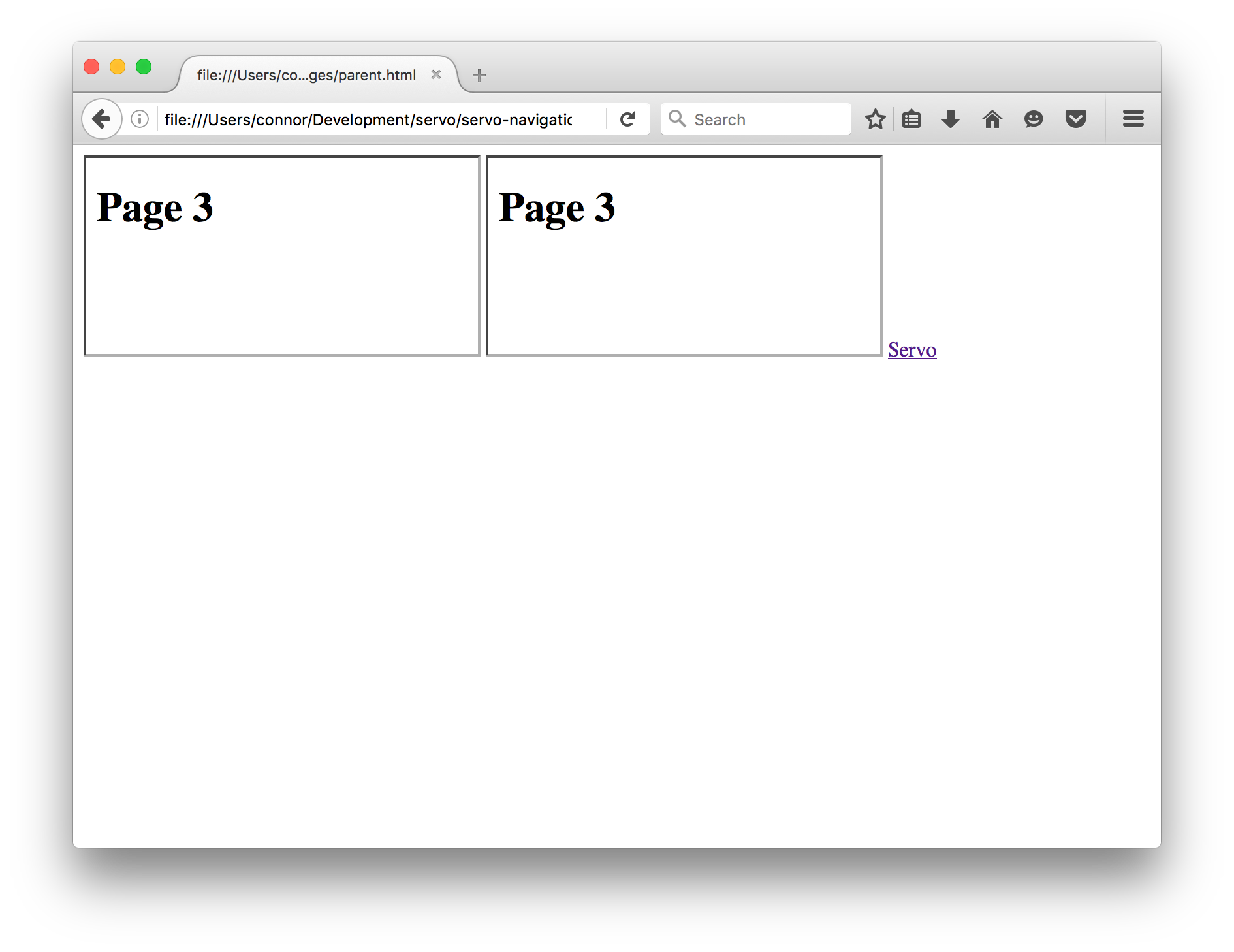}
    }~\raisebox{-.5\height}{\rlap{
      \begin{tikzpicture}
        \node[doc,active,fully](0) at (0,0){0};
        \node[doc](1) at (1,-1){1};
        \node[doc](2) at (2,-2){2};
        \node[doc](3) at (3,-1){3};
        \node[doc,active,fully](4) at (4,-1){4};
        \node[doc](5) at (5,-2){5};
        \node[doc,jshactive,fully](6) at (6,-2){6};
        \node[draw,dotted,fit=(0)]{};
        \node[draw,dotted,fit=(1)(4)]{};
        \node[draw,dotted,fit=(2)(6)]{};
        \draw[->](0)to[out=0,in=140](4);
        \draw[->](0)to[out=0,in=120](6);
      \end{tikzpicture}
    }}
\end{quote}
and traversing by $-4$ results in state:
\begin{quote}
    \raisebox{-.5\height}{
      \includegraphics[width=.45\linewidth]{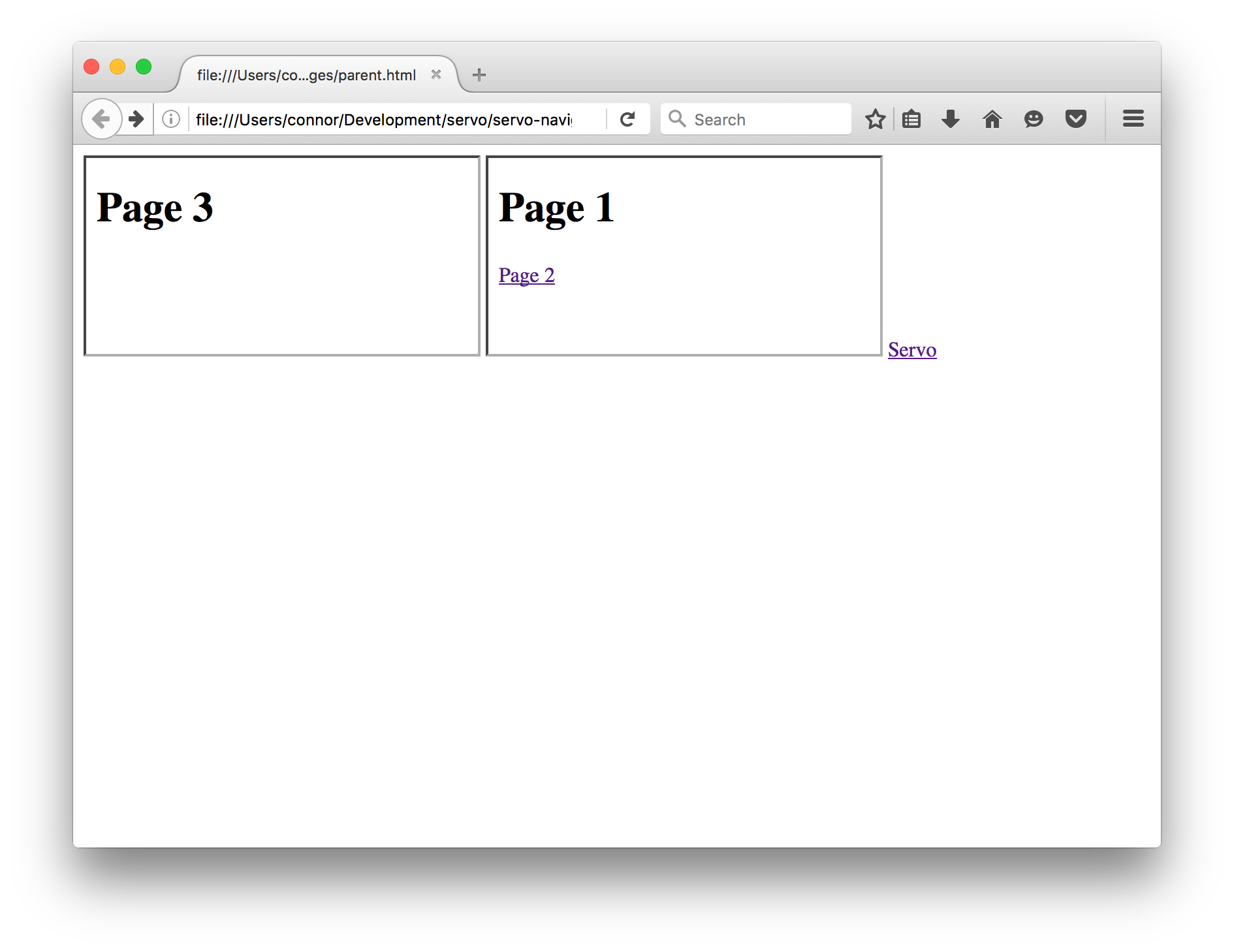}
    }~\raisebox{-.5\height}{\rlap{
      \begin{tikzpicture}
        \node[doc,active,fully](0) at (0,0){0};
        \node[doc](1) at (1,-1){1};
        \node[doc,jshactive,fully](2) at (2,-2){2};
        \node[doc](3) at (3,-1){3};
        \node[doc,active,fully](4) at (4,-1){4};
        \node[doc](5) at (5,-2){5};
        \node[doc](6) at (6,-2){6};
        \node[draw,dotted,fit=(0)]{};
        \node[draw,dotted,fit=(1)(4)]{};
        \node[draw,dotted,fit=(2)(6)]{};
        \draw[->](0)to[out=0,in=140](4);
        \draw[->](0)to[out=-20,in=90](2);
      \end{tikzpicture}
    }}
\end{quote}
This state is unexpected, as document $4$ should have traversed to document $1$, and any state
showing \verb|page3.html| should be capable of going back.

In Safari, the use of \verb|pushState| and \verb|popState| for navigation has unexpected
results. We can use \verb|pushState| and \verb|popState| to construct state:
\begin{quote}
    \raisebox{-.5\height}{
      \includegraphics[width=.45\linewidth]{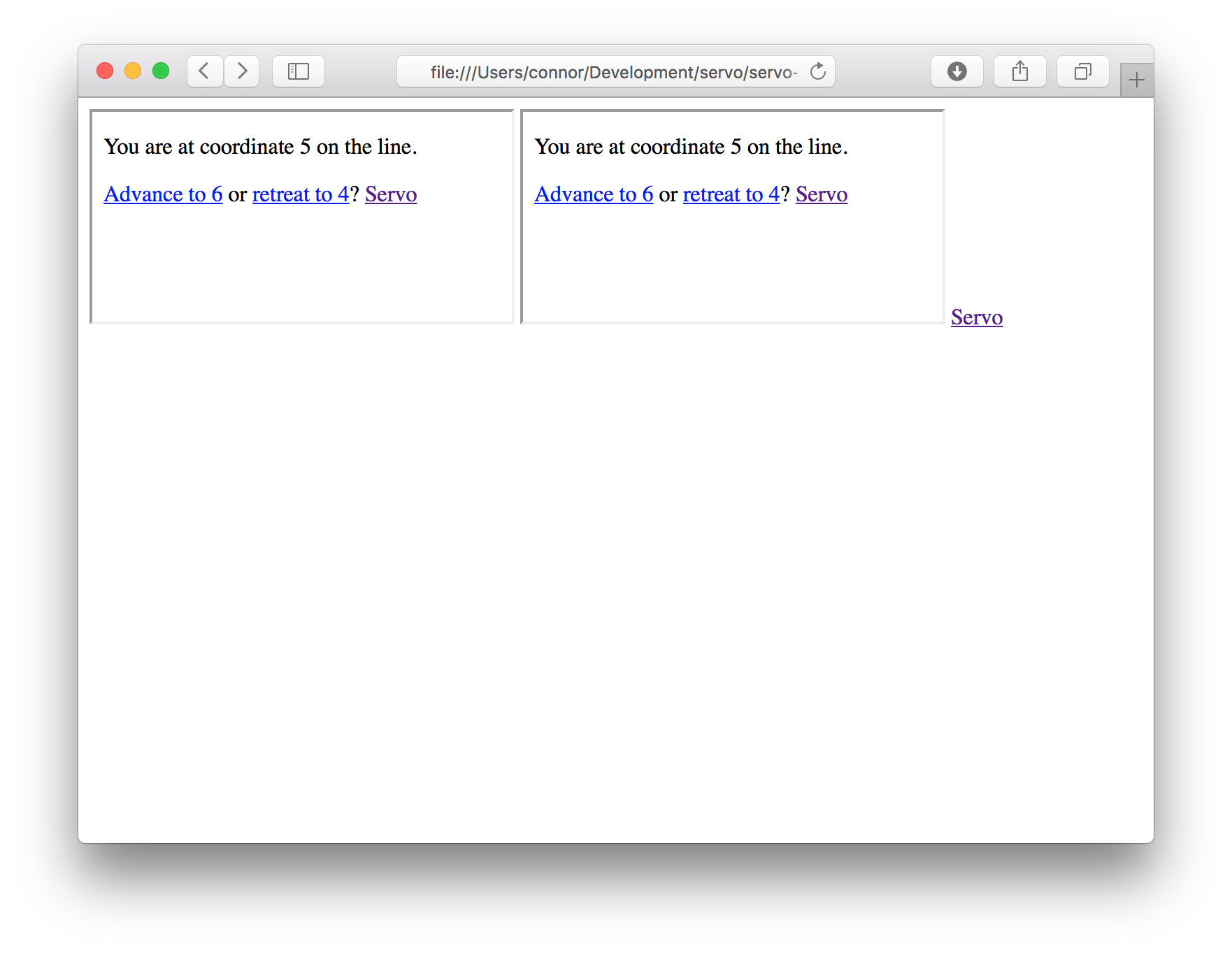}%
    }~\raisebox{-.5\height}{\rlap{
      \begin{tikzpicture}
        \node[doc,active,fully](0) at (0,0){0};
        \node[doc,active,fully](1) at (1,-1){1};
        \node[doc,jshactive,fully](2) at (2,-2){2};
        \node[doc](3) at (3,-1){3};
        \node[doc](4) at (4,-1){4};
        \node[doc](5) at (5,-2){5};
        \node[doc](6) at (6,-2){6};
        \node[draw,dotted,fit=(0)]{};
        \node[draw,dotted,fit=(1)(4)]{};
        \node[draw,dotted,fit=(2)(6)]{};
        \draw[->](0)--(1);
        \draw[->](0)to[out=-20,in=90](2);
      \end{tikzpicture}
    }}
\end{quote}
then traversing by $+4$ results in:
\begin{quote}
    \raisebox{-.5\height}{
      \includegraphics[width=.45\linewidth]{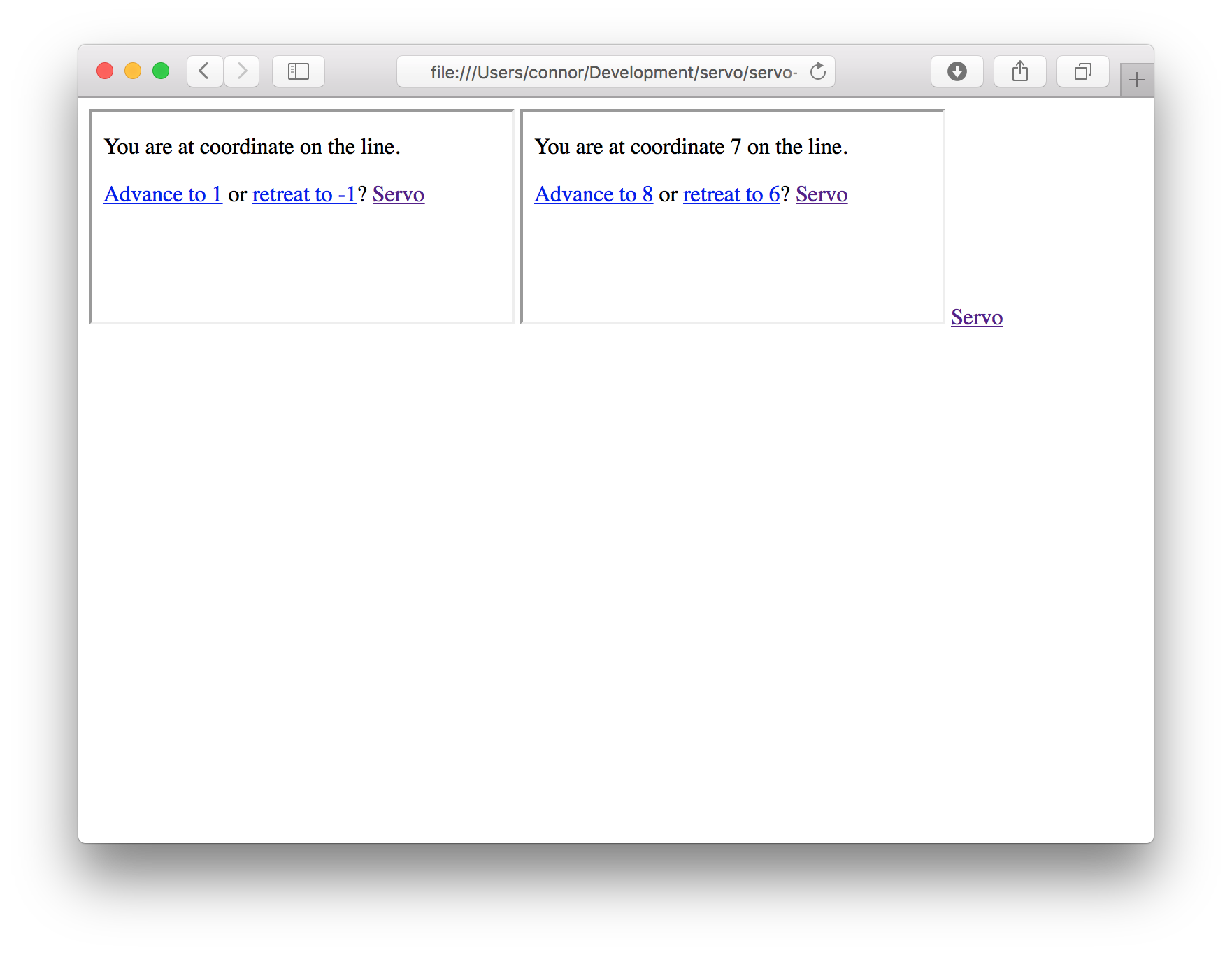}%
    }~\raisebox{-.5\height}{\rlap{
      \begin{tikzpicture}
        \node[doc,active,fully](0) at (0,0){0};
        \node[doc](1) at (1,-1){1};
        \node[doc](2) at (2,-2){2};
        \node[doc](3) at (3,-1){3};
        \node[doc,jshactive,fully](4) at (4,-1){4};
        \node[doc](5) at (5,-2){5};
        \node[doc](6) at (6,-2){6};
        \node[draw,dotted,fit=(0)]{};
        \node[draw,dotted,fit=(1)(4)]{};
        \node[draw,dotted,fit=(2)(6)]{};
        \draw[->](0)to[out=0,in=140](4);
      \end{tikzpicture}
    }}
\end{quote}
After this traversal, we are unable to determine the active entry for one of the \verb|iframe|s
as its state is \verb|null|.

As these examples show, navigation history is difficult to implement: even major
browser implementations give unexpected behaviours when combining separate
\verb|iframe| session histories.

\section{Specification}

In this section, we discuss how the \textsc{whatwg}
specification~\cite[\S7.7.2]{whatwg} can be aligned with the model
from \S\ref{sec:model}. This is not a direct translation, due to some
of the features we elided in our model. In particular, we did not
discuss how documents are \emph{loaded} and \emph{unloaded}, which
includes downloading and allocating resources such as \textsc{html} or
\textsc{css}, and activating JavaScript content. Since
loading-then-unloading a document is wasteful, the specification
should be written to avoid loading intermediate pages when traversing
by a delta. This introduces complexity.

Our first proposed change is that the current specification is defined in terms
of the ``joint session history'' and makes use of the ``current entry
of the joint session history'', neither of which are used by our model.
We propose to remove the definition of ``joint session history''
and ``current entry of the joint session history'', and add the following:
\begin{quote}
  The \textbf{session past} of a browsing context is the entries
  of the session history added before the current entry
  (and does not include the current entry).

  The \textbf{session future} of a browsing context is the entries
  of the session history added after the current entry
  (and does not include the current entry).

  If an entry has a next entry in the chronologically ordered session
  history, it is its \textbf{successor}.

  If an entry has a previous entry in the chronologically ordered session
  history, it is its \textbf{predecessor}.

  The \textbf{joint session past} of a top-level browsing context is the
  union of all the session pasts of all browsing contexts
  that share that top-level browsing context.

  Entries in the joint session past are in decreasing chronological order of
  the time they were added to their respective session histories.

  The \textbf{joint session future} of a top-level browsing context is the
  union of all the session futures of all browsing contexts
  that share that top-level browsing context.

  Entries in the joint session future are in increasing chronological order of
  the time their predecessor were added to their respective session
  histories.
\end{quote}
The second proposed change is to replace the definition of how a user agent
should``traverse the history by a delta'' by the following:
\begin{quote}
  To \textbf{traverse the history by a delta} \emph{delta}, the user
  agent must append a task to this top-level browsing context's session
  history traversal queue, the task consisting of running the following
  steps:
  \begin{enumerate}
  \item Define the \emph{entry sequence}
    as follows:
    \begin{enumerate}

    \item If \emph{delta} is a positive integer $+n$, and if the length of the
      joint session future is less than or equal to $n$ then let the \emph{entry sequence}
      be the first $n$ entries of the joint session future.

    \item If \emph{delta} is a negative integer $-n$, and if the length of the
      joint session past is less than or equal to $n$ then let the \emph{entry sequence}
      be the first $n$ entries of the joint session past.

    \item Otherwise, abort traversing the history by a delta.

    \end{enumerate}

  \item A session entry is said to \textbf{become active} when
    it is a member of the \emph{entry sequence}, and no
    session entry after it in the \emph{entry sequence} has the same
    browsing context.

  \item A session entry is said to \textbf{stay active} when it it the
    current entry of its browsing context, and there are no members of
    the \emph{entry sequence} with the same browsing context.

  \item A session entry is said to be \textbf{activating} when either
    it will become active or stay active.

    \textbf{Note:} the activating documents
    will be active after traversal has finished.

  \item A session entry is said to be \textbf{fully activating} if
    is activating, and either its browsing context is a top-level
    browsing context, or it has a parent browsing context
    and the session entry through which it is nested is itself fully activating.

    \textbf{Note:} the fully activating documents
    will be fully active after traversal has finished.

  \item Queue a task that consists of running the following
    substeps. The relevant event loop is that of the specified
    browsing context's active document. The task source for the queued
    task is the history traversal task source.

    \begin{enumerate}

    \item For each \emph{specified entry} in the \emph{entry sequence},
      run the following substeps.
      \begin{enumerate}

      \item Let \emph{specified browsing context} be the browsing context of the \emph{specified entry}.

      \item If there is an ongoing attempt to navigate \emph{specified
        browsing context} that has not yet matured (i.e. it has not
        passed the point of making its \texttt{Document} the active
        document), then cancel that attempt to navigate the browsing
        context.

      \item If the \emph{specified browsing context}'s active document
        is not the same \texttt{Document} as the \texttt{Document} of
        the specified entry, then run these substeps:
        \begin{enumerate}

        \item Prompt to unload the active document of the
          \emph{specified browsing context}. If the user refused to
          allow the document to be unloaded, then abort these steps.

        \item Unload the active document of the \emph{specified
          browsing context} with the recycle parameter set to false.

        \end{enumerate}

      \item If the \emph{specified entry} is activating but not fully activating,
        then set the current entry of the session history of \emph{specified browsing context}
        to be the \emph{specified entry}.

        \textbf{Note:} in this case, the current entry of the session history should
        be updated, but the document will not be fully active, so should not be loaded.

      \item If the \emph{specified entry} is fully activating, then
        traverse the history of the \emph{specified browsing context}
        to the \emph{specified entry}.

        \textbf{Note:} in this case, the document will be fully active, so should be loaded.

      \end{enumerate}

    \end{enumerate}

  \end{enumerate}
\end{quote}
We believe that these changes bring the specification in line with our model, and
so satisfies the fundamental property of navigation.

\section{Conclusion}

We have proposed a model of web navigation compatible with the
\textsc{whatwg} specification, and investigated its ``fundamental
property'': that traversing by $\delta$ then by $\delta'$ is the same
as traversing by $\delta+\delta'$.  Unfortunately, the specified model
does not satisfy this property, but we have shown that a patched model
does. Experimentally, it appears that the patched model is closer to
the behaviour of existing browser implementations.

\bibliographystyle{plain}
\bibliography{notes}

\end{document}